\documentclass[12pt]{article}
 \usepackage{fullpage}
 \usepackage{amssymb,amsmath,amsfonts,amsthm}
\usepackage{graphicx}
 \usepackage{cite}
\usepackage{url}

 \usepackage{color}
 \usepackage{hyperref}
 \hypersetup{colorlinks=true,
     linkcolor=blue,          
     citecolor=blue,        
     filecolor=blue,         
     urlcolor=blue }

 \newtheorem{theorem}{Theorem}
 \newtheorem{lemma}{Lemma}

\newcommand{\no}[1]{}
\newcommand{\todo}[1]{} 

\newcommand{\eps}{\varepsilon}

\usepackage{enumitem}
\setlist[itemize]{noitemsep, topsep=0pt} 
\setlistdepth{9}
\newlist{myenumerate}{enumerate}{9}
\setlist[myenumerate,1]{label=(\arabic*)}
\setlist[myenumerate,2]{label=(\Roman*)}
\setlist[myenumerate,3]{label=(\Alph*)}
\setlist[myenumerate,4]{label=(\roman*)}
\setlist[myenumerate,5]{label=(\alph*)}
\setlist[myenumerate,6]{label=(\arabic*)}
\setlist[myenumerate,7]{label=(\Roman*)}
\setlist[myenumerate,8]{label=(\Alph*)}
\setlist[myenumerate,9]{label=(\roman*)}

\newcommand{\old}{\mathrm{old}}
\newcommand{\new}{\mathrm{new}}

\title{External-memory dictionaries with worst-case update cost} 

\author{
Rathish Das\thanks{University of Liverpool, UK. Supported by the Canada Research Chairs Programme and NSERC Discovery Grants.}
\and
John Iacono\thanks{Universit\'{e} libre de Bruxelles, Belgium and New York University, USA. Supported by Fonds de la Recherche Scientifique-FNRS under Grants MISU F 6001 1 and CDR/OL J.0101.22 --- 40008322.}
\and
Yakov Nekrich\thanks{Michigan Technological University, USA.}
}
\date{}

\begin{document}


\maketitle

 \begin{abstract}
The $B^{\epsilon}$-tree [Brodal and Fagerberg 2003] is a simple I/O-efficient external-memory-model data structure that supports updates orders of magnitude faster than B-tree with a query performance comparable to the B-tree: for any positive constant $\epsilon<1$ insertions and deletions take $O(\frac{1}{B^{1-\epsilon}}\log_{B}N)$ time (rather than $O(\log_BN)$ time for the classic B-tree), queries take $O(\log_BN)$ time and range queries returning $k$ items take $O(\log_BN+\frac{k}{B})$ time.
Although the $B^{\epsilon}$-tree has an optimal update/query tradeoff, the runtimes are amortized.
Another structure, the write-optimized skip list, introduced by Bender et al.~[PODS 2017], has the same performance as the $B^{\epsilon}$-tree but with runtimes that are randomized rather than amortized.
In this paper, we present a variant of the $B^{\epsilon}$-tree with deterministic worst-case running times that are identical to the original's amortized running times.
 \end{abstract}

\section{Introduction}

The external memory model of Aggarwal and Vitter~\cite{AggarwalVi88}, sometimes called the I/O model or the Disk-Access Model (DAM), has been the most successful model of computation for problems where the data can not fit into RAM and the transfers between main memory and external memory strongly dominate the runtime on real machines.
External memory data structures~\cite{BrodalFa03b,BrodalFa02b, Arge03, ArgeBeDe07,buchsbaum2000external, jiang2019faster, ONeilChGa96,BenderFaFi07, SearsCaBr08,BenderFaJo17,DBLP:journals/algorithmica/DemaineIL15, iacono2012using,MunroN19} have played a significant role in improving the performance of applications involving large datasets; the above references are but a small sample of related work in the external memory model, DBLP currently lists \href{https://dblp.org/search?q=external+memory}{over 400 publications} with external memory in the title.



In the external memory model, there are two levels of memory: an internal memory of size $M$ and an external memory of unbounded size connected to the internal memory. Data is transferred between the two levels of memory in contiguous blocks of size $B$. The cost of an algorithm or a data structure is measured by the number of block transfers between internal and external memory; all computation in internal memory is free. When designing an algorithm in the external-memory model, the values of $M$ and $B$ are known, in contrast to the cache-oblivious model~\cite{DBLP:journals/talg/FrigoLPR12}.


The $B$-tree, introduced by Bayer and McCreight fifty years ago~\cite{BayerMc72}, is the archetypal data structure in the external memory model. It stores a set of totally ordered keys, and supports insertions, deletions, predecessor queries in time $O(\log_BN)$, and range searches returning $k$ elements in time $O(\log_BN+\frac{k}{B})$ in the external memory model. The $B$-tree is simple, is a standard part of the CS curriculum, and in optimized form is widely implemented as the cornerstone structure of databases due to its excellent real-world performance. 

In internal memory, updates and searches share the same $O(\log N)$ runtime which is achieved by classic structures such as the AVL tree and the Red-Black tree. However, in external memory it is possible to substantially speed up updates with little increase in the search time.
This trade-off was explored primary from the lower-bound view by Brodal and Fagerberg~\cite{BrodalFa03b}. On one extreme of this trade-off curve, B-trees have an optimal query (e.g., predecessor query) bound but have a slow update bound. On the other extreme of the trade-off curve, buffer-repository trees~\cite{buchsbaum2000external} have far better update bound than B-tree but have a poor query bound. 

Brodal and Fagerberg~\cite{BrodalFa03b} introduced $B^{\epsilon}$-tree at a ``sweet spot'' on the update-query trade-off curve. $B^{\epsilon}$-tree performs insertions orders of magnitude faster than B-tree with a query performance comparable to B-tree. The improvement in insertion performance of $B^{\epsilon}$-tree comes due to the use of buffers at nodes. Instead of inserting an individual key into the tree, the key is queued in the buffer of the root node, and when a significant number of keys are buffered, they are flushed recursively to the next level of the tree. This buffering technique allows $B^{\epsilon}$-tree to achieve an amortized update cost $O(\frac{1}{B^{1-\epsilon}} \log_B N)$, while slowing down the queries by but a $\frac{1}{\epsilon}$ factor ($\epsilon$ is a tuning parameter that is between 0 and 1); they showed this is optimal for a random insertion workload~\cite{BrodalFa03b}. 
Data structures for other fundamental abstract data types (ADTs) in the external memory and cache oblivious models with fast updates have also been created; see for example cache-oblivious dictionaries \cite{DBLP:conf/soda/BrodalDFILM10},
priority queues \cite{DBLP:conf/esa/IaconoJT19,FJKT99,KS96,ABDHM07,BFMZ04,CR18},
hashing (only exact search) \cite{iacono2012using} and
point location \cite{DBLP:conf/isaac/IaconoKK19}.

The $B^{\epsilon}$-tree structure is very simple and easy to implement, but it has one major flaw: the runtimes hold in the amortized sense only and in the worst case the entire structure may need to be modified to execute a single query. Such performance is clearly unacceptable in the target application of a database. As we describe below, the $B^{\epsilon}$-tree is just a $B$-tree with fanout $B^{\epsilon}$, and where there is a buffer on each internal node which caches updates and only distributes them to the children when the buffer is full.
The deterministic worst-case running time guarantee is very poor due to \emph{flushing cascades}. A flushing cascade occurs when flushing the buffer of a node triggers flushes into multiple children nodes, which in turn trigger flushes to their children and so on. 
Flushing cascades become most acute when two nodes are merged into a new node, and their buffers are merged into a single buffer. The resulting buffer can overflow if the merged buffer can not hold all the items of the original two buffers. A buffer overflow can trigger new flushes, which again cause cascades and more node merges in the whole tree. In the worst-case, a single update (insert or delete) could trigger modifications to $\Omega(N^{1-o(1)})$ nodes of the tree~\cite{bender2020flushing}.

There have been several attempts to deterministically de-amortize the performance of $B^{\epsilon}$ trees. In Bender et al.~\cite{BenderFaJo17}, a variant of skip lists was presented where a query takes $O(\log_{B} N)$ I/O with high probability and an update takes $O(\frac{1}{B^{1-\eps}}\log_{B} N )$ I/O amortized with high  probability.
Recently Bender et al.~\cite{bender2020flushing}  introduced a randomized universal buffer flushing strategy that achieves an update bound of $O(\frac{1}{B^{1-\eps}}\log_{B} N )$ with high probability (not amortized) for a variant of the $B^{\epsilon}$-tree. 

So, using the best previously known results, $O(\frac{1}{B^{1-\epsilon}} \log_B N)$ updates and $O(\log_BN)$ queries are possible in either the amortized sense, or randomized with high probability.

\textbf{This paper.}
Our contribution is that these running times can be achieved deterministically in the worst-case (i.e., deamortized) with a structure based on the $B^{\epsilon}$-tree. 

It seems difficult to apply standard de-amortization techniques to our problem without sacrificing performance. 
The standard global re-building
technique~\cite{overmars1987design} works by copying insertions into a new tree using a background process. When all insertions are copied into a new tree, the old tree is discarded.  This approach can guarantee that the height of the tree is bounded by $O(\log_B N)$.
However, we can have a large sequence of leaves with a very small number of insertions in each leaf. This can significantly increase the cost of range queries; in the worst scenario we may have to visit $\Omega(N/B)$ leaf nodes in order to answer a range query.
Another approach is to maintain the sizes of buffers in internal nodes using a background process. However, the merging of two nodes can lead to a cascade of deletions. Thus each round of the background process could take almost linear time as explained above~\cite{bender2020flushing}. 

In this paper we  de-amortize the $B^{\eps}$-tree through a combination of several techniques. At a high level, this involves having large leaves stored in a separate structure, and periodically splitting or merging selected leaves rather than when they reach certain sizes. Splitting and merging needs to be propagated up the structure, and splitting in particular can cause overflow which requires flushing data down the structure. And, all of this needs to be done while new updates keep arriving though a very careful choice of relevant parameters.

What does it mean to deamortize a structure whose per-operation runtime may be subconstant, as it is likely to be in the case of $O(\frac{1}{B^{1-\eps}}\log_{B} N )$? This means that if the update cost of $\frac{1}{B^{1-\eps}}\log_{B} N $ is subconstant, then a constant number of I/Os are executed every $\frac{B^{1-\epsilon}}{\log_B N}$ operations.

We proceed by reviewing the standard $B^\epsilon$-tree in Section~\ref{s:btree} and presenting our variant in Section~\ref{s:ods}.


\section[The B\^epsilon]{The $B^{\epsilon}$-tree}\label{s:btree}
\begin{figure}[!ht]
\centering
\includegraphics[width=0.7\textwidth]{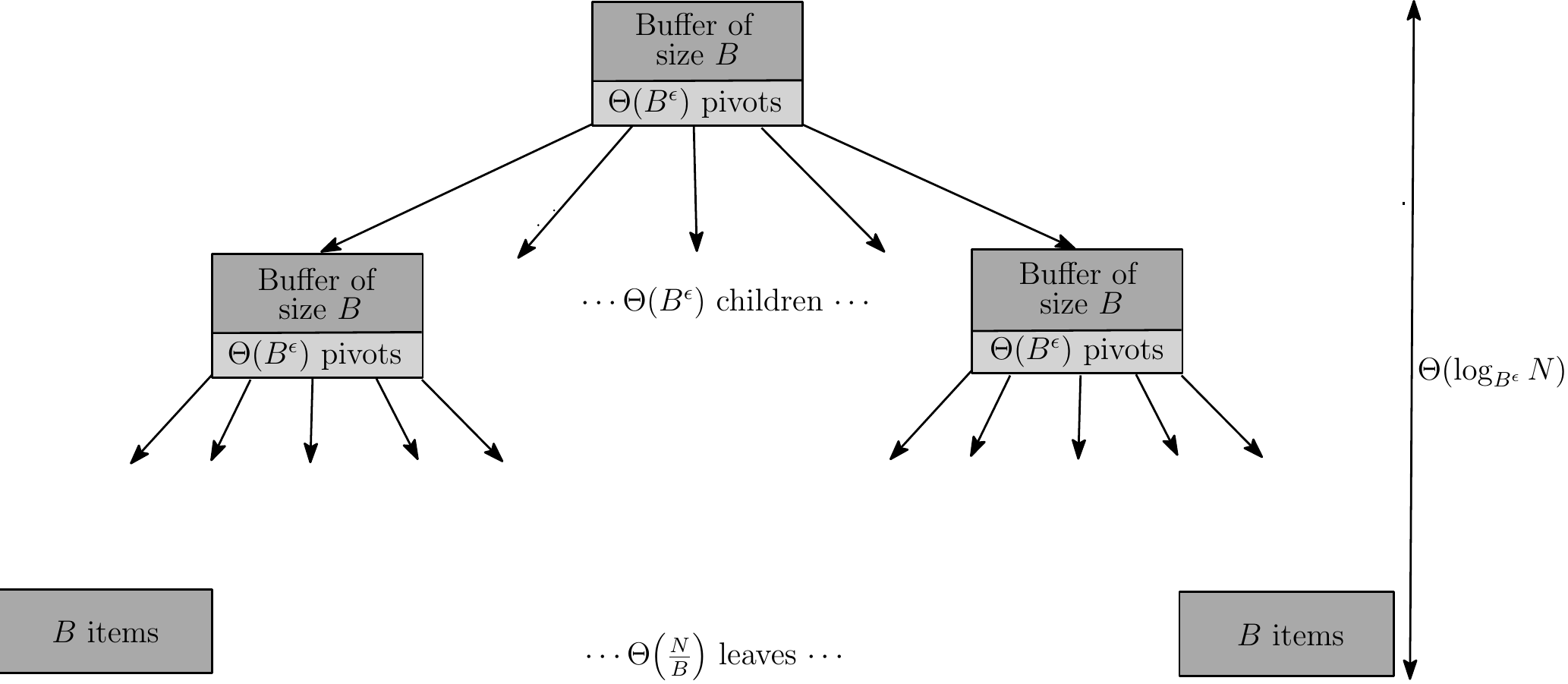}
\caption{$B^{\epsilon}$-tree.}
\label{fig:tree}
\end{figure}

\paragraph*{Abstract data type.}
The data structure maintains a set $S$ of keys (or key-value pairs), where the keys are constant-sized and come from a totally ordered universe, under the following operations:
\begin{itemize}
    \item Insert($x$): Adds key $x$ to $S$ under the precondition that $x \not \in S$.
    \item Delete($x$): Removes key $x$ from $S$ under the precondition that $x \in S$
    \item Predecessor($x$): Returns the largest key $y \in S$ such that $y \leq x$
    \item Range-Report($x,y$): Returns all keys in $[x,y] \cap S$ in sorted order.
\end{itemize}
We use $N$ to denote $|S|$, and use the word \emph{update} to generically refer to an insertion or deletion.


\paragraph*{Structure description.} See Figure~\ref{fig:tree}.
A $B^{\epsilon}$-tree is just a standard $B$-tree where internal nodes have $\Theta(B^{\epsilon})$ children and an additional buffer of size $B$ which stores updates that are destined to the subtree of the node; leaves store $\Theta(B)$ keys. The buffers of internal nodes store updates, that is, both insertions and deletions. 

Predecessor queries can be done with a single root-to-leaf traversal, pushing updates on the search path down to the leaf. During this process, paired insertions and deletions annihilate each other. In the case where such annihilations require restructuring via standard splits and merges, it is paid for in the amortized sense with the removal of the updates from the tree.

When a buffer is full, some of the updates are removed and copied to the buffer of a child, with the child being chosen so as to maximize the number of updates pushed down; as the buffer is size $B$, the keys are constant sized, and the number of children is $\Theta(B^{\epsilon})$, there must be at least one child buffer with $\Omega(B^{1-
\epsilon})$ updates to be pushed down to it. Determining this child and moving the data can be done in a constant number of block transfers.
When leaves are full or become more than a constant factor empty, splits and merges are performed in the usual way, though merges may cause overflows of merged buffers which could cause more updates to be propagated, which could cause more merges, and this could continue as possibly effect every node in the structure.
However, the amortized cost of an insertion/deletion is $O(\frac{1}{B^{1-\epsilon}}\log_B N)$ as each update will be involved in $O(\log_BN)$ buffer overflows at a unit cost of each, spread over $B^{1-
\epsilon}$ items. We summarize the performance of $B^\epsilon$ trees:

\begin{theorem}[\cite{BrodalFa03b}]
Given any constant  $\eps$, $0<\eps<1$, the $B^\epsilon$-tree supports updates in amortized time $O(\frac{1}{ B^{1-\eps}}\log_B N)$, predecessor searches in worst-case time $O(\log_B N)$, and range searches that return $k$ elements in worst-case time $O(\log_B N+\frac{k}{B})$ in the external memory model.
\end{theorem}


\section{Our Data Structure}
\label{s:ods}

Out data structures performance is summarized by the following theorem, which we will spend the rest of this section proving:

\begin{theorem}
Given any constant  $\eps$, $0<\eps<1$, our data structure supports updates in worst case time $O(\frac{1}{ B^{1-\eps}}\log_B N)$ 
(meaning $O(1)$ I/Os every $O({ B^{1-\eps}}{\log_B N)}$ updates if $\frac{1}{ B^{1-\eps}}\log_B N<1$), predecessor searches in worst-case time $O(\log_B N)$, and range searches that return $k$ elements in worst-case time $O(\log_B N+\frac{k}{B})$ in the external memory model.
\end{theorem}

\subsection{Our Approach}
We make several major changes in our structure compared to a standard buffered $B^{\epsilon}$-tree in order to obtain worst-case bounds. We describe these changes at a high level first, and then give a full description of our data structure.
 Following the convention used in many previous papers, our data  structure stores \emph{updates}. An update is a key combined with a flag indicating whether it is an insertion or a deletion. To avoid confusion, we use update \emph{operation} to refer to a insertion or deletion operation executed by the user, and \emph{update} without the word ``operation'' to refer to a key/flag pair stored by the data structure.
\label{subsec:approach}
 We define $\delta=\eps/2$. We assume $N$ is an upper bound on the number of keys logically stored by the structure, and will state several invariants of the structure in terms of $N$; should the actual number of keys logically stored vary polynomially from from $N$, the structure will need to be rebuilt, which we explain how to do in \S\ref{sec:rebuild}.
 We make the following major changes in our structure compared to the standard $B^{\epsilon}$-tree:
\begin{itemize}
\item
  We increase the size of leaf nodes to $\Theta(B\log^2_B N)$. Each leaf is stored in an auxiliary data structure implemented as a two-level tree, with fan-out $\Theta(\log_B N)$.
\item
  Each internal node stores a buffer of size at most $ B^{1-\delta}$. The only exception is we allow the root and one other node to have a buffer size double that, $2B^{1-\delta}$, because two buffers were just combined as a result of a merge.
\item
Instead of splitting and merging leaves when they reach some size thresholds, we alternate picking the largest leaf and splitting it (if it is large enough) and picking the smallest leaf and merging it with a sibling (if it is small enough), and then possibly splitting it. Such splits and merges are propagated up in the normal way.
We show that this is sufficient to maintain the size invariants of the leaves.
\item
To ensure that a node's buffer size is within range, either because it is the root and updates have been added as the result of update operations or a merge caused the two sibling buffers to be combined,
$B^{1-2\delta}$ items are moved from a buffer to a single child (\emph{a flush}), and this flush is repeated down to a leaf. Multiple such flushes, up to $B^\delta$ may be required to reduce a merged buffer of size $2B^{1-\delta}$ down to $B^{1-\delta}$.
\item We maintain auxiliary information in each internal node. Such information includes the key values needed for searching, which is standard, as well as augmented information that allows the algorithm to find the leaves storing the most/least updates. When changes are made to leaves, the auxiliary information may need to be updated in its ancestors.
\item We alternate waiting for the user to execute $\Theta(\frac{B^{1-2\delta}}{\log_BN})$ update operations (as well as an unlimited number of queries), and doing structural work to maintain the invariants until there is an I/O. The various parameters are chosen with care so that the number of update operations does not overwhelm the rebalancing in progress.
  \item Queries do not change the structure. Instead, we guarantee by construction that all updates found in buffers of internal nodes on a root-to-leaf path fit in $\Theta(\log_B N)$ blocks and that the insertions on the path asymptotically dominate the deletions. Thus all of these updates can be copied into a temporary workspace, all insert-delete pairs are removed, the result of the predecessor query is determined and returned. Range queries work in a similar manner, with the insertions contained in the buffers that cover any range dominating the deletions. This approach avoids the need to answer queries by pushing down the changes to the root and causing a possibly uncontrolled number of cascading changes as in the classic $B^\epsilon$-tree.
\end{itemize}

Summing up, our tree is leaf-heavy, among the updates stored in the buffers insertions asymptotically dominate deletions, most updates are stored in the leaves, and each leaf stores $\Theta(B\log_B ^2 N)$ updates. The tree structure is maintained by regularly splitting and merging the leaves that are largest and smallest and not at some fixed threshold of size.

\subsection{Invariants}

As in a standard $B^\epsilon$ tree, internal nodes store buffers of updates that are destined for some leaf in their subtreee, and keys to guide the search. The leaves, however, are separate structures described in \S\ref{sec:leaf}.

\begin{description}
\item[Child count.] Every internal node has a number of children in the range $[\frac{1}{2}B^\delta,B^\delta]$. The root may have less children, but at least 2. This invariant can be violated by a single node at any given time, and may only violate it by $\pm 1$. 
\item[Leaf size.] Every leaf stores a number of updates in the range $[ B \log^2_BN,5 B \log^2_BN]$. We do not count matching insertion/deletion pairs in a leaf as these will be eliminated.

\item[Internal node buffer size.] Each internal node has at most $B^{1-\delta}$ updates in its buffer. This invariant can be violated by two nodes at any given time. One of which is the root, and the other is a node that will have just had its buffer merged with a sibling. For these two nodes we ensure that there are at most $2B^{1-\delta}$ updates in their buffers. These nodes are said to be \emph{overfull}.

\item[Height.] The height of the tree is always at most $c_h \log_BN$ for some constant $c_h$. This follows from the preceding invariants.

\item[Auxiliary information.] As is standard in a leaf oriented-tree, each internal node stores keys needed to guide searches. Additionally, each internal node stores which child contains the subtree with the leaf of maximum size, and what that size is, as well as the same information for the leaf of minimum size.

%

\end{description}

\subsection{Updates}

First we describe at a high level how update operations are executed, then present in Figure~\ref{alg} the details, followed by justification of correctness and runtime.

One core subroutine is the \emph{flush} which moves $B^{1-2\delta}$ updates from a node's overfull buffer to one of its children, and then repeats this process on that child until either a non-overfull buffer or a leaf is reached. If a buffer is overfull it has at least $B^{1-\delta}$ updates in its buffer. Thus, by pigeonhole there must be one of its at most $B^\delta$ children for which it can move $B^{1-2\delta}$ updates to.

At a high level, our algorithm can be thought of always being in the process of doing one of the following:

\begin{itemize}
    \item While the root buffer is overfull, flush. This could require $\Theta(B^\delta)$ flushes.
    \item The largest leaf is split if it stores at least $4 \log^2_B N$ updates. This split is propagated up the tree.
    \item The smallest leaf is merged if it stores less than $2 \log^2_B N$ updates. This merge is propagated up the tree. However, each merge may cause a buffer to become too large, up to $2B^{1-\delta}$ instead of the normal limit of $B^{1-\delta}$; this is fixed by performing $O(B^\delta)$ flushes on the overfull buffer to remove the overfull condition. It is crucial to note that $O(B^\delta)$ flushes could be performed starting on \emph{each} of the $O(\log_BN)$ nodes being split; as the flush is itself a loop that could visit $O(\log_BN)$ nodes, this creates a triple-nested loop whose innermost operations, one step of a flush, could run $\Theta(B^\delta \log^2_BN)$ times.
\end{itemize}

One cycle will consist of either splitting or merging a leaf and propagating upwards, where after each time a node is split or merged the root is flushed. Figure~\ref{alg} shows the full details of this.

The crucial point is that after every I/O in the above infinite process the algorithm stops and waits for the user to execute $\frac{B^{1-2\delta}}{c_i\log_BN}$ update operations, which it adds to the root's buffer, before continuing; the constant $c_i$ is defined later.

We now argue the claimed runtime is correct, the claimed number of times each loop runs is correct, and the invariants hold. Arguing that the leaf size invariant holds is the most involved part of the result. 


\begin{figure} \small
\newcounter{lno}
\newcommand{\labelline}[1]{\item[\refstepcounter{lno}\label{#1}{\texttt{\ [\ref{#1}]}}]}
\begin{myenumerate}
    \labelline{l1} Alternate between a \emph{split phase} and a \emph{merge phase}. Together, these are a \emph{cycle}.
    \begin{myenumerate}
        \labelline{l3} Set $n_1$ to the root. Move $n_1$ to the leaf  with the most/least updates if a split/merge phase by repeatedly moving $n_1$ down to the appropriate child using the auxiliary information. (This takes $c_h \log_B N$ steps).
            \labelline{l4}
        If a split phase and $n_1$'s buffer has at least $4 B \log_B^2N$, split $n_1$; if a merge phase and $n_1$ has at most $2B \log_B^2N$ updates, merge $n_1$ with one of its siblings and if this results in a leaf with more than $5 B \log_B^2 N$ updates split it evenly. 
        \labelline{l5} \textbf{Propagate splits/merges up.} While $n_1$ is not null: (the parent of the root) (This loop will run at most $c_h \log_B N$ times).
        \begin{myenumerate}
            \labelline{l6} Move $n_1$ to its parent.
            \labelline{l7} Split/merge $n_1$ as in a $B$-tree if needed; also split or merge the buffer(s). Update the auxiliary information in $n_1$.
            \labelline{l8} Set $n_2$ to $n_1$.
            \labelline{l9} \textbf{Repeated flushing of split/merge node.} While $n_2$ is an internal node with an overfull buffer: (This loop will run at most $B^\delta$ times).
            \begin{myenumerate}
                \labelline{l10} Set $n_3$ to $n_2$.
                \labelline{l11} \textbf{Propagating a flush of a split/merge node downward.} While $n_3$ is an internal node with $n_3$ an overfull buffer: (This loop will run  at most $c_h \log_BN$ times).
                \begin{myenumerate}
                        \labelline{l12} Identify a child of $n_3$, call it $n_4$, for which the buffer of $n_3$ can move at least $B^{1-2\delta}$ elements from $n_3$'s buffer to $n_4$'s, this exists by pigeonhole.
                        \labelline{l13} Move $B^{1-2\delta}$ items from $n_3$'s buffer to $n_4$'s.
                        \labelline{l14} Set $n_3=n_4$.
        \end{myenumerate}
    \labelline{l16} \textbf{Update auxiliary information.} Move $n_3$ to its parent until it reaches the root while updating the auxiliary information. (This loop runs at most $B^\delta$ times).

            \end{myenumerate}

    \labelline{l18} While the root buffer  is overfull, set $n_5$ to the root and: (This loop runs at most $B^\delta$ times). 
    \begin{myenumerate}
    \labelline{l19} \textbf{Flush from root.} While the buffer of $n_5$ is overfull: (This loop runs at most $c_h\log_BN$ times).
    \begin{myenumerate}
                        \labelline{l20} Identify a child of $n_5$, call it $n_6$, for which the buffer of $n_5$ can move at least $B^{1-2\delta}$ elements from $n_5$'s buffer to $n_6$'s; this exists by pigeonhole.
                        \labelline{l21} Move $B^{1-2\delta}$ items from $n_5$'s buffer to $n_6$'s.
                        \labelline{l22}  Set $n_5=n_6$.
    \end{myenumerate}
    \labelline{l24} \textbf{Update auxiliary information.} Moving $n_5$ to its parent until it reaches the root while updating the auxiliary information.
    \end{myenumerate}    \end{myenumerate}
           \end{myenumerate}
    \end{myenumerate}
    \caption{Full description of how update operations are executed. Note that this is presented as an infinite loop. In addition to the above description, every time there is an I/O the algorithm stops and waits for the user to execute $\frac{B^{1-2\delta}}{c_i\log_BN}$ update operations and then contininues until the next I/O occurs.} \label{alg}
\end{figure}

\paragraph*{The leaf data structure: a preview.}

We include here a summary of the leaf operations and their performance so that we may analyze the main update procedure using the leaves as a black box. Section~\ref{sec:leaf} describes how the leaves are implemented to achieve this.
The leaf data structure will support the following subroutines; as in the description of the main procedure to execute update operations, these algorithms alternate execution and waiting for the user to execute update operations.

\begin{description}
\item[Bulk-insert of updates.]
$B^{1-2\delta}$ updates are added to the leaf. During the execution of this $O(\log_BN)$ I/Os are performed, with the algorithm stopping after each one to wait for the user to execute $\frac{B^{1-2\delta}}{c_i\log_BN}$ update operations which are added to the root buffer of the main tree.
\item[Split and merge.]
These take $O(B^\delta \log_BN)$ I/Os to complete; the algorithm stops after each I/O to wait for the user to execute $\frac{B^{1-2\delta}}{c_i\log_BN}$ update operations which are added to the root buffer of the main tree.
\end{description}

\paragraph*{Runtime.} We require that after each I/O  we wait for $\frac{B^{1-\epsilon}}{c_i\log_BN}$ update operations to be executed by the user (recall $2\delta=\epsilon)$. This is exactly what having a worst case runtime of $O(\frac{\log_BN}{B^{1-\epsilon}})$ means, when $\frac{1}{B^{1-\epsilon}}<1$, a constant number of I/Os every $\frac{B^{1-\epsilon}}{c_i\log_BN}$ updates. 

We have presented the algorithm and analysis for the (common) case where $\frac{B^{1-\epsilon}}{c_i\log_BN}$ is at least one. If, however, we are in the (uncommon) case where $\frac{B^{1-\epsilon}}{c_i\log_BN}$ is less than one, instead wait for the user to execute a single update operation and add it to the root buffer every $\frac{c_i\log_BN}{B^{1-2\delta}}$ I/Os. This trivially gives the claimed runtime.

\paragraph*{Loops.} Here we argue that each loop runs the number of times indicated. The number of times each loop runs, except those of lines \ref{l9} and \ref{l18} of Figure~\ref{alg}, are clearly bounded by the height of the structure as they involve going up or down. However, the repeated flushing in  lines \ref{l9} and \ref{l18} requires a different argument. If one of these loops run, it is because a flush is happening at one of the two nodes that may violate size of an internal node invariant, either the root (line~\ref{l18}) or a node that had just had it buffer merged with another (line~\ref{l11}). These nodes may have up to $2B^{1-\delta}$ updates in their buffers. Each execution of these loops will remove $B^{1-2\delta}$ elements from the node being flushed. This will require looping at most $B^\delta$ times to restore the invariant that they have at most $B^{1-\delta}$ updates in their buffer.

\paragraph*{Children.} As in a standard $B$ or $B^\epsilon$ tree, the number of children may be out of range, and only by one, while in the process of propagating a split or merge up; this node where the splits and merges are in progress is $n_1$ in our presentation.

\paragraph*{Internal node buffer size.}

In each execution of the loop at line~\ref{l5}, the total number of I/Os is at most $c_i  B^\delta \log_B N$, for some $c_i$. This is where $c_i$ is defined. Since $\frac{B^{1-2\delta}}{c_i\log_B N}$ update operations are executed per I/O by adding the updates to the root, this means 
$B^{1-\delta}$ updates are added to root per execution of the loop at line~\ref{l5}. In execution of the loop at line~\ref{l5} we explicitly flush from the root in line~\ref{l18} until the internal buffer size is not longer overfull.

The only other time an internal node's buffer becomes overfull is after a merge, but in line~\ref{l11} flushes are performed until the node is no longer overfull.

\paragraph*{Leaf size.}
In order to show that the leaf size invariant is maintained, that is each leaf holds from $  B\log^2_B N$ to $5  B\log^2_B N$ updates, we will need the following result which comes from Theorem 5 of \cite{DietzS87}. Our formulation of this theorem will allow its easy use later.

\begin{theorem}
  \label{DietzS87}
  Let $X$ be a set of real valued positive variables and let $b$ be a positive constant. All variables $x\in X$ are initially $0$. We execute the following process which proceeds in rounds and in each round may change the values of the $x\in X$ as follows:
  \begin{enumerate}
    \item Values may be decreased without restriction
    \item The sum of the increases in $x$ of those $x\in X$ that increase in a round is at most $b$.
    \item During each round, the maximum value of $x$ must be set to zero at some point.
  \end{enumerate}
  Additionally, new elements may be added to $X$ with a value of 0 and elements with a value of 0 may be removed from $X$ at any time.
  If $|X|$ is always bounded by $m$,
  there is a constant $c_t$ such that at all times $x_j \leq c_t b \log m$ for all $j$.
\end{theorem}

\begin{lemma}
  \label{lemma:backgr}
    The leaf size invariant is always satisfied, that is, the number of updates in each leaf node is always between $ B \log^2_BN$ and $5 B \log^2_BN$.
\end{lemma}

\begin{proof}
To avoid repetitive clutter, let $\tau$ denote $B \log^2_BN$.

We examine how the number of updates in each leaf can change. These changes are caused by the flushing of updates into the leaves, splitting leaves, and merging leaves. We discuss each of these separately.

Let $L$ be the set of leaves, and let $n(\ell)$ be the number of insertions stored in a leaf $\ell$ without a matching delete. Define $d(\ell)$ be $\max(0, n(\ell)- 4 \tau )$, and
$d'(\ell)$ be $\max(0, \tau -d(\ell) )$. Thus, $d$ is nonzero when a leaf is within $\tau$ of its upper limit and $d'$ is nonzero when a leaf is within $\tau$ of its lower limit.

In a cycle, updates can be added to leaves when a flush reaches a leaf, this could happen in lines~\ref{l13} and \ref{l21} of Figure~\ref{alg}. Multiplying by the number of times these lines may be executed in a cycle, the number of updates that are added to all leaves 
is 
at most $c_a B^{1-\delta} \log_B N$ for some constant $c_a$.
This can cause $\sum_{\ell \in L} d(\ell)$ to increase by that much during a cycle as insertions are added to a leaf.
When an deletion is added to a leaf, this will cause $n(\ell)$ to decrease by 1.
Thus $\sum_{\ell \in L} d'(\ell)$ can also increase by $c_a B^{1-\delta} \log_B N$ per cycle.

Additionally, any split will cause the created leaves during the split to have $d(\ell)$ and $d'(\ell)$  be  zero as they will have a size of at least $4 \tau$ and at most $5 \tau$ before the split, which results in a size of at least $ 2 \tau$ and at most $\frac{5}{2} \tau $.

If no split is performed then $n(\ell)<4 \tau $ for all leaves $\ell$ and thus $d(\ell)$ is zero for all leaves $\ell$.

If a merge is performed, 
a node with $n(\ell)$ in the range $[\tau ..2 \tau]$ merges with its sibling which has $n(\ell)$ in the range $[\tau ..5 \tau]$, resulting a new leaf with size in the range $[2\tau..6\tau]$. If its size is in the range $[2\tau..4\tau]$ it has a $d(\ell)$ and $d'(\ell)$ value of 0. Otherwise its size is in the range $[4\tau..6\tau]$ and it is split evenly into two nodes each of which has size $[2\tau..3\tau]$ and thus has $d(\ell)$ and $d'(\ell)$ be $0$.

Thus in each cycle we have shown that $\sum_{\ell \in L} d(\ell)$ and 
$\sum_{\ell \in L} d'(\ell)$ each increase by at most $c_a B^{1-\delta} \log_B N$ and
the leaf $\ell$ with maximum $d(\ell)$ and $\ell'$ with maximum $d'(\ell)$ will both have their values $d$ and $d'$ reset to zero if a split/merge was performed, and if it was not, they were zero already.
Thus we can apply Theorem~\ref{DietzS87} twice: one where the elements of $X$ represent the $d(\ell)$ values, $m$ represents the number of leaves (which is trivially at most $N$) and $b$ represents 
$c_a B^{1-\delta} \log_B N$, and once using $d'(\ell)$ instead.

This yields bounds on $\max_\ell d(\ell)$ and $\max_\ell d'(\ell)$ of
$(c_a B^{1-\delta} \log_B N) \cdot c_t \log N$, which is at most $B \log^2_BN$, for large enough $B$.

The bounds on $d(\ell)$ and $d'(\ell)$ being at most $\tau$ immediately follows by the definition of $d$ and $d'$ that $n(\ell)$ is in the range $[\tau,5\tau]$ which, recalling that $\tau=B \log^2_BN$, is the claim of the Lemma.
\end{proof}

\subsection{Leaf Data Structures} 
\label{sec:leaf}
We now describe how the leaf data structures are implemented. 
Each leaf data structure supports adding $B^{1-2\delta}$ updates into the leaf, splits, merges, and searches. The main structure ensures that splits and merges are performed to maintain the leaf's size invariant.

We store each leaf $\ell$ as a two-level tree, and use the terminology \emph{micro-root} and \emph{micro-leaf} to avoid confusion with the main structure. 
Each leaf $T(\ell)$ consists of a micro-root with $\Theta(\log_B N)$ child micro-leaves.  
Each micro-leaf of $T(\ell)$ contains $\Theta(B\log_B N)$ updates in sorted order. 
The micro-root of $T(\ell)$ has a buffer with at most $B^{1-\delta}\log_B N$ updates. 
Every time when we flush $B^{1-2\delta}$ updates into $\ell$ we add them to the root buffer of $T(\ell)$.  
If the root buffer contains over $B^{1-\delta}\log_B N$ updates, we identify the micro-leaf  $\nu$ where at least $B^{1-2\delta}$ updates can be moved. 
We flush those updates to $\nu$ and add them to blocks of $\nu$. 
We can merge the newly flushed updates with the extant updates in $\nu$ in $O(\log_B N)$ I/Os, while maintaining the updates in sorted order, and annihilating matching insertion/deletion pairs on the same element.
A deletion that is added to a micro-leaf is mutually annihilated with the insertion of the same element, which must be present in the leaf. 
Thus micro-leaves contain insertions only. Since deletions can be present only at the buffer of the micro-root of $T_\ell$, there can be at most $B^{1-\delta}\log_B N$ deletions in any leaf node.    
Micro-leaves of $T(\ell)$ can be split and merged in a standard way so that each micro-leaf holds $\Theta(B \log_B N)$ updates. 

When a leaf $\ell$ is split into $\ell_1$ and $\ell_2$, we distribute the micro-leaves of $T(\ell)$ among $T(\ell_1)$ and $T(\ell_2)$; one micro-leaf can be possibly split into two parts. The cost is $O(\log_B N)$ I/Os.  When two leaves $\ell_1$ and $\ell_2$ are merged, all micro-leaves of $T(\ell_1)$ and $T(\ell_2)$ become the micro-leaves of $T(\ell)$. The buffers of the micro-roots of $T(\ell_1)$ and $T(\ell_2)$ are merged too. The cost of merging is $O(\log_B N)$ I/Os. After merging, the micro-root buffer of $T(\ell)$ can contain up to $2B^{1-\delta}\log_B N$ I/Os. We ``repair'' the micro-root buffer by moving excessive updates to micro-leaves. We can flush $B^{1-2\delta}$ updates to a micro-leaf in $O(\log_B N)$ I/Os. Hence the total cost of flushing all excessive updates to micro-leaves is $O(B^{\delta}\log_B N)$.  

In all of these operations, every time there is an I/O the algorithm stop and waits for $\frac{B^{1-2\delta}}{c_i\log_BN}$ update operations to be executed by the user, adds these operations to the buffer of the root of the main tree, and then resumes.

\subsection{Queries}
\label{sec:query}

To answer a range reporting query $[a,b]$, we identify all leaves that intersect with the range $[a,b]$. If a leaf $\ell$  intersects with but is not contained in $[a,b]$, we identify the micro-leaves that intersect with $[a,b]$. If a micro-leaf $\ell_m$ intersects with $[a,b]$, we traverse $\ell_m$ and make a list of insertions stored in $[a,b]$. If a micro-leaf is entirely contained in $[a,b]$, we list  all its insertions. Let $L(\ell)$ denote the list of insertions found in $\ell$.   We examine updates stored in the root buffer of $T(\ell$) and modify $L(\ell)$ accordingly. Since the updates in the root buffer are in sorted order, $L(\ell)$ can be modified in $O(\log_B N +|L(\ell)|/B)$ I/Os, where $|L|$ denotes the number of updates in a list $L$.  
If a leaf $\ell$ is entirely contained in $[a,b]$, we  make a list $L(\ell)$ containing all  insertions stored in  micro-leaves and in the root buffer of $T(\ell)$ (minus the deletions in the root buffer of $T(\ell)$). This takes $O(|L(\ell)|/B)$ I/Os.  There are at most two leaves that partially intersect with the query range $[a,b]$. Hence total cost of generating lists $L(\ell)$ for all relevant leaves $\ell$ is $O(k/B + \log_B N)$, where $k$ is the number of insertions that must be reported.  We then visit  all ancestors of relevant leaves $\ell$, add insertions from $[a,b]$ to an extra list $L_a$, and copy all updates from $[a,b]$ to internal memory.  Finally we report all elements from $L(\ell)$, excluding the deletions.  Assuming that the internal memory can hold $O(\log_B N \cdot B^{1-\delta})$ deletions, this can be done in $O(k/B)$ I/Os.  Thus we can answer reporting queries in $O(k/B +\log_B N)$ I/Os under assumption that the internal memory is not too small, $M> B\log_B N$. 
We note that $B^\epsilon$ trees as well as the structures of Bender et al.~and Bender et al.~\cite{bender2020flushing,BenderFaJo17} all similarly require small non-constant amounts of memory for queries.

To answer a predecessor query $q$ (i.e., to find the largest element that is not larger than $q$), we create a list $L_1$ of  $B\log_B N$ largest insertions stored in the leaf nodes that are $\le q$. This can be done in $O(\log_B N)$ I/Os using  the same reporting procedure as described above. Since a leaf holds $\Theta(B\log^2_B N)$ insertions,  insertions in $L$ are from at most  two consecutive leaves. We examine all insertions stored in ancestors of these two leaves and make a list $L_2$ of all insertions $\le q$. Let $L_3$ denote the list of $B\log_B N$ largest elements in $L_2$ and $L_3$.  Finally we  make the list $L_d$ that consists of all deletions $\le q$  stored in  ancestors. We find the largest element in  $L_3$ that is not in $L_d$ using the following lemma.
\begin{lemma}
  \label{lemma:arrays}
  Let $X$ and $Y$ be two sets such that $|X|=2m$, $|Y|=m$, and $Y\subset X$. We can find the largest element $e\in X$, such that $e\not\in Y$ in $O(m/B)$ I/Os. 
\end{lemma}
\begin{proof}
  Let $e(X,r)$ and $e(Y,r)$ denote the $r$-th largest elements in $X$ and $Y$ respectively. Let $X(r_1,r_2)$ be the set of elements $e(X,r_1)$, $\ldots$, $e(X,r_2)$.  The cost of finding  $e(X,r)$ and $e(Y,r)$ for any $r$ is   $O(\frac{m}{B})$ using linear-time selection. 

  If $e(X,k)=e(Y,k)$, then the $k$ largest elements in $X$ and $Y$ are identical. If $e(X,k)> e(Y,k)$, then at least one among the $k$ largest elements of $X$ does not occur in $Y$. We must  find the smallest $k'$, such that $e(X,k')>e(Y,k')$. Clearly $e(X,k')$ is the required element.   If $e(X,m)=e(Y,m)$, then $k'=m+1$.  Otherwise $e(X,m)> e(Y,m)$ and we proceed as follows. We compare $e(X,m/2)$ and $e(Y,m/2)$.  If $e(X,m/2)=e(Y,m/2)$, we recursively search for $k'$ among elements of $X(m/2, m)$ and $Y(m/2,m)$. If $e(X,m/2)>e(Y,m/2)$, we recursively search for $k'$ among the elements of $X(1, m/2-1)$ and $Y(1,m/2-1)$.  The total runtime is $O(m/B)$ I/Os. 
\end{proof}

Lists $L_1$, $L_2$, $L_3$, and $L_d$ can be constructed in $O(\log_B N)$ I/Os. We can find the largest insertion $x$ in $L_3$ that is not in $L_d$ in $O(|L_3|/B)=O(\log_B N)$: first we extract from $L_d$ all deletions that do not exceed $q$ and then we apply Lemma~\ref{lemma:arrays}. Clearly, $x$ is the largest element that is not deleted and is not larger than $q$. Hence $x$ is the predecessor of $q$ and can be found in $O(\log_B N)$ I/Os. 

Summing up, we can answer predecessor and membership queries in $O(\log_B N)$ I/Os. The cost of answering a reporting query is $O(\log_B N + k/B)$ where $k$ is the number of reported values.  For reporting queries (but not for membership and predecessor queries), we need to make an assumption that the internal memory can hold at least $\log_B N$ blocks. However  the same assumption is also necessary in the case of standard $B^{\eps}$-trees with amortized updates.

\subsection{Rebuilding} \label{sec:rebuild}
In our description we assumed that the values of  $\log_B N$ and $\tau$ are fixed.  However, if the size of the data structure changes significantly (e.g., the number of updates in the data structure changes from $N_0$ to below $\sqrt{N_0}$), the value of $\tau$ can change by more than a constant factor. Hence the invariant of leaf size, that is the $\Theta(B\log^2_BN)$ bound on leaf size can be violated.

In order to maintain the invariant of leaf size, we must re-build the tree in certain situations. Let $N_0$ denote the number of updates in the tree when it was re-built for the last time. If $N\ge (1/2)N_0^2$ or $N\le (3/2)\sqrt{N_0}$, we set $\tau'=B\log_B^2N$ and construct a new tree, such that all leaves hold  between $(1/4)\tau'$  and $(7/4)\tau'$ updates. The new tree $T^{\new}$ can be constructed in the background as explained below.

Let $T^{\old}$ denote the currently used tree and let $T^{\new}$ denote the new constructed tree. 
We traverse the updates stored in the leaf nodes of $T^{\old}$ in left-to-right order. Updates stored in each leaf are examined in increasing order (this is easy to do because the updates in micro-leaves and in the root node buffer of $T({\ell})$ are sorted).  Let $\ell$ denote the currently visited leaf node and let $L$ denote the list of insertions stored in the ancestors of $\ell$.   During each round we examine the next $2B^{1-\delta}$ insertions stored in $\ell$. We add these insertions to $L$ and extract the set $S$ of $B^{1-\delta}$ smallest insertions. Next, we visit all ancestors of $\ell$ and  remove from $S$ all deletions (that is, if $S$ contains an insertion $e$ and an ancestor of $\ell$ stores a deletion $e$, then $e$ is removed from $S$). We also remove all deletions stored in the root buffer of $T(\ell)$. All insertions that remain in $S$ are added  into the new tree $T^{\new}$. When all insertions in a leaf $\ell$ are examined, we move to the right neighbor of $\ell$ and update the list $L$ accordingly. 

When new updates are  added to $T^{\old}$, we also add them to $T^{\new}$. The only exception are deletions that correspond to unprocessed insertions:  Let $v_{\max}$ denote the key value of the largest insertion from $T^{\old}$ that was added to $T^{\new}$. All deletions with key value $v\le v_{\max}$  are added to $T^{\new}$; deletions with key value $v> v_{\max}$ are not added to $T^{\new}$.  When all insertions in the leaves of $T^{\old}$ are examined and copied to $T^{\new}$, we discard $T^{\old}$ and start using $T^{\new}$ to answer queries.

\emph{\textbf{Remark}}: To simplify the description, we set $\delta=\eps/2$ in our data structure. Hence the height of our tree is two times larger than the height of the standard  $B^{\eps}$-tree.  It is possible to use any constant
$\delta<\eps$ such that $B^{\delta}\log B< B^{\eps}$. Thus it is possible to reduce the height ratio of our tree and the standard $B^{\eps}$-tree to any constant $\rho>1$.

\bibliography{worst-case-EM}


\end{document}